\date{\today}

\documentclass[jctcce,letterpaper,twocolumn,floatfix,preprintnumbers,superscriptaddress]{revtex4}

\usepackage{dcolumn,graphicx,amsmath,amssymb,algorithm,algpseudocode}
\usepackage{graphicx}
\usepackage{amsthm}
\usepackage{mathtools}
\usepackage{todonotes}
\usepackage[space]{grffile}
\usepackage{subfig}
\newtheorem{theorem}{Theorem}

\usepackage{caption}
\DeclarePairedDelimiter{\floor}{\lfloor}{\rfloor}

\input{Qcircuit.tex}

\begin{document}

\title{Low Rank Density Matrix Evolution for Noisy Quantum Circuits}
\author{Yi-Ting Chen}
\email{yitchen@stanford.edu}

\affiliation{QC Ware Corp., Palo Alto, CA 94301}
\affiliation{Department of Applied Physics, Stanford University, Stanford, CA 94305}

\author{Collin Farquhar}
\email{cjf235@cornell.edu}

\author{Robert M. Parrish}
\email{rob.parrish@qcware.com}
\affiliation{QC Ware Corp., Palo Alto, CA 94301}

\begin{abstract}
     In this work, we present an efficient rank-compression approach for the classical simulation of Kraus decoherence channels in noisy quantum circuits. The approximation is achieved through iterative compression of the density matrix based on its leading eigenbasis during each simulation step without the need to store, manipulate, or diagonalize the full matrix. We implement this algorithm in an in-house simulator, and show that the low rank algorithm speeds up simulations by more than two orders of magnitude over an existing implementation of full rank simulator, and with negligible error in the target noise and final observables. Finally, we demonstrate the utility of the low rank method as applied to representative problems of interest by using the algorithm to speed-up noisy simulations of Grover's search algorithm and quantum chemistry solvers.
\end{abstract}

\maketitle


\section{Introduction}
The scaling of quantum computers is limited by quantum decoherence \cite{Pellizzari1995,Chuang1995,Copsey2003,Ashhab2006}. State of the art quantum computers that consist of fewer than 100 qubits \cite{Arute2019,Gomes2018} are of interest for noisy intermediate-scale quantum (NISQ) applications, where qubits are used without error correction \cite{Preskill2018}. Therefore, classically simulating imperfect and noisy circuits is vital for the development and design of NISQ-era algorithms as well as characterizing errors in quantum hardware \cite{Harper2020}. Existing classical simulators have typically focused on emulating noiseless circuits. In this case, simulations of quantum circuits with more than 50 qubits has been demonstrated \cite{Pednault2017,Chen2018,Dang2019}. There are several techniques developed for speeding up simulations of certain types of circuits \cite{Bravyi2016,Jozsa2014,Vidal2003,Plesch2010,Bartlett2002} or algorithms \cite{Yoran2007,Browne2007,Shi2006,Kassal2008}. For simulation of noisy circuits, there are high performance computations developed \cite{Khammassi2017,Wei2018,Chaudhary2019,Jones2019}, as well as light-weight open source tools such as density matrix simulators in Qiskit \cite{Qiskit2019} and Cirq \cite{Cirq2019}. These simulators are based on evolving full density matrices which is prohibitively expensive for large numbers of qubits.  

In an open quantum system, the decoherence can be modeled as interactions with a large environment \cite{Breuer2007}. One can determine the properties of an open quantum system with the Monte Carlo wave function method where a system is decomposed into an ensemble of pure states that evolve individually and then are averaged \cite{Dalibard1992,Molmer1993,Bassi2008,Guerreschi2020,Abid2020}. On the other hand, the dynamics can also be characterized by the Lindblad equation \cite{Gorini1976,Lindblad1976} which well describes decoherence in various quantum hardware architectures \cite{Jelezko2004,Raitzsch2009,Fitzpatrick2017}. To reduce the complexity of the Lindblad equation, one can project the quantum states onto a lower dimensional basis using filtering theory and simulate the states more efficiently with reasonable accuracy \cite{Handel2005,Bris2013}. 

In this work, we combine the ideas of the pure state decomposition \cite{Dalibard1992,Molmer1993} and the low dimension basis projection \cite{Bris2013} to efficiently simulate noisy quantum circuits. Compressed representations are commonly applied to classical simulation of quantum systems with high symmetry or low entanglement. Applications range from compressed sensing of quantum state tomography \cite{Gross2010,Kyrillidis2018}, limiting bond dimensions of tensor networks \cite{Vidal2003,Vidal2004,White2018}, and low-rank factorization of Hamiltonians \cite{Motta2018} to efficiently represent of states \cite{Cao2010,Plesch2010,Wu2018,Wu2018-2}. In our case, it is found that the von Neumann entropy of a density matrix is often small when the noise level is low, implying that it is possible to model the density matrix using a matrix of lower rank with minimal information loss. We achieve this by iteratively projecting onto a subspace of the eigenbasis, and evolving only a small ensemble of pure states. 

In the following, we present a complete and explicit algorithm which decomposes a mixed density matrix into a low rank matrix representing an ensemble of pure states, applies gate and Kraus operators to this low rank matrix, and computes the output density matrix and probability distribution. The procedure involves iterative compression of the density matrix to maintain the most numerically compact form with minimal error. As an example, Fig. \ref{fig:demo_low_error}a shows a 6-qubits density matrix after a quantum circuit that solves a Grover's search problem for finding states with Hamming weight $\leq2$. The same circuit is simulated with depolarizing noise with noise strength $p\simeq 0.33\%$ by an exact method and by our low rank method. In this example, we use a low rank representation that has only $20\%$ of the full rank. Fig. \ref{fig:demo_low_error}b and \ref{fig:demo_low_error}c show that the low rank method simulates noise with high accuracy. More extensive benchmarking and detailed descriptions on the performance and accuracy of the method are in Section \ref{sec:implementation}. 

In fact, we show that it is possible to evolve and compute quantities of interest of a $(2^N \times 2^N)$ density matrix without ever forming a matrix of size $(2^N \times 2^N)$, where $N$ is the number of qubits. The algorithm is then assessed by a sequence of random benchmarking under various types and strengths of noise channels to test its practical speed-up and error. We show that the algorithm performs consistently in random circuits and in structured circuits for quantum algorithms such as Grover's search, with a speed-up more than two orders of magnitude and with a small error (around $0.01\%$) in the probability distribution associated with the final output density matrix. Furthermore, as $N$ becomes larger, and approaches the range for which classical simulations become difficult, the advantage of this algorithm continues to increase over the standard method of full density matrix evolution.

\begin{figure*}[]
  \includegraphics[width=1.0\textwidth]{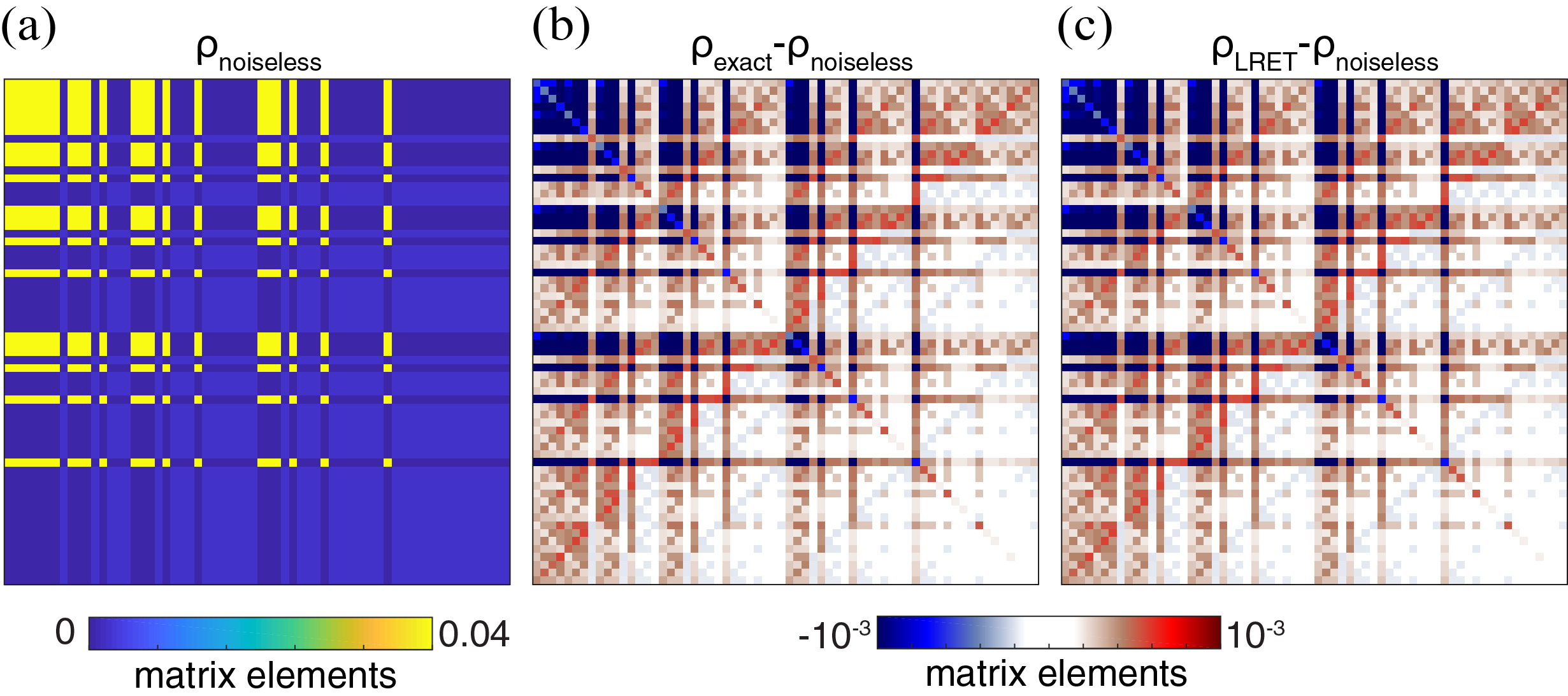}
  \caption{An example showing that the low rank method simulates noise with negligible error. (a) density matrix from noiseless simulation, $\rho_\text{noiseless}$. (b) difference of density matrices from exact noise simulation and from noiseless simulation, $\rho_\text{exact}-\rho_\text{noiseless}$. (c) difference of density matrices from low rank noise simulation and from noiseless simulation, $\rho_\text{LRET}-\rho_\text{noiseless}$ where LRET, the focus of this article, is a low rank method. The rank of $\rho_\text{LRET}$ is 20\% of that of $\rho_\text{exact}$, and corresponds to 2.2\% of distortion (defined in Section \ref{sec:implementation}). The noise strength is $p\simeq 0.33\%$ and the truncation threshold is $\epsilon=3\times 10^{-4}$.}
  \label{fig:demo_low_error}
\end{figure*}

\section{Algorithm for Low Rank Noise Simulation}
In this section, we present an algorithm that simulates noisy circuits using a low rank representation of density matrices. The algorithm consists of two parts, low rank evolution and eigenvalue truncation, which are covered in section \ref{sec:low rank evolution} and \ref{sec:eigenvalue truncation} below. In section \ref{sec:kraus operation decomposition}, an iterative procedure consisting of these two parts is introduced. Then, in section \ref{sec:probability density and measurement} we explain how to sample the associated probability distribution without explicitly forming the full density matrix.

\subsection{Low Rank Evolution} \label{sec:low rank evolution}
A coherent quantum system can be represented by either a statevector $\ket{\psi}$ or a density matrix $\rho=\ket{\psi}\bra{\psi}$. $\ket{\psi}$ has dimension $2^N\times1$, while $\rho$ has dimensions $2^N\times 2^N$. Because of the substantial difference in the sizes of $\ket{\psi}$ and $\rho$, most classical simulations of coherent quantum systems work directly with the statevector representation. Unfortunately, the statevector representation does not directly allow for the presence of decoherence. Instead, the evolution of a quantum system in a decoherent noise channel can only be described by a density matrix $\rho$, which is generally more computationally demanding to simulate. 
The evolution of a quantum state in a noisy quantum circuit is described by \cite{Kraus1983,Bacon2001,Nielsen2011}
\begin{equation} \label{eqn:full rank evolution}
\rho^{(d+1)} = \sum_{\alpha} p_{\alpha} K_{\alpha} \rho^{(d)} K_{\alpha}^{\dag}
\end{equation}
where $K_\alpha$ are Kraus matrices, $p_{\alpha}$ is the corresponding probability, $\rho^{(d)}$ and $\rho^{(d+1)}$ are density matrices before and after a noise channel. This Kraus-based noise model is capable of encapuslating many different types of decoherence channels, such as bit flip, depolarizing, etc., and therefore is an extremely useful tool in modeling the operation of NISQ-era quantum algorithms in the presence of decoherence noise on real devices. However, most current implementations of this Kraus-based decoherence model explicitly work with the $2^N \times 2^N$ noise-including density matrix. Our approach exploits the fact that for realistically small noise levels ($p\simeq0.01$), the von Neuman entropy of the density matrix usually remains low, raising the possibility of working with an approximate but accurate low-rank representation of the density matrix. This observation is supported by an entropy analysis in Appendix \ref{adx:LowRankEntropy}.

While the formal possibility of a low-entropy density matrix evolution is tantalizing, there remains to be resolved many pragmatical details about how to efficiently identify and exploit this rank structure while avoiding formation and manipulation of any density-matrix-sized quantities in the rank identification process. Here and in II.B we describe an algorithm that can accomplish this for noise-including density matrices that exhibit the desired low-rank structure. A density matrix, $\rho$, can always be decomposed as a outer product of the $L \in \mathbb{C}^{2^N\times V}$ matrix,
\begin{equation} \label{eqn:low rank decomposition}
	\rho \equiv L L^{\dag},
\end{equation}
for some $V\in \mathbb{N}$. While the choice of $L$ is not unique, in general, it is possible to find $L$ with $V$ equal to the rank of the density matrix using decomposition methods such as singular value decomposition. For density matrices with rank smaller than $2^N$, this form most compactly represents the state with the minimal column dimension. Using the decomposition in Eq. (\ref{eqn:low rank decomposition}), we can evolve the density matrix by updating $L$ without evaluating $\rho^{(d)}$ explicitly as in Eq.(\ref{eqn:full rank evolution}). For a gate operation
\begin{equation} \label{eqn:low rank evolution G}
    \begin{split}
        \rho^{(d+1)} 
        = \mathcal{G}^{(d)} \rho^{(d)}
        & = G^{(d)} \rho^{(d)} G^{(d)\dag} \\
        & = G^{(d)} L^{(d)} L^{(d)\dag} G^{(d)\dag} \\
        & = L^{(d+1)} L^{(d+1)\dag}
    \end{split}
\end{equation}
where $L^{(d+1)}\equiv G^{(d)} L^{(d)}$, $\mathcal{G}^{(d)}$ is a gate operation, and $G^{(d)}$ is its corresponding gate matrix. Likewise, for a Kraus operator,
\begin{equation} \label{eqn:low rank evolution K}
    \begin{split}
        \rho^{(d+1)}
        & = \mathcal{K}^{(d)} \rho^{(d)}
        = \sum_{\alpha=1}^A p_\alpha K_\alpha^{(d)} \rho^{(d)} K^{(d)\dag}_\alpha \\
        & = \sum_{\alpha=1}^A p_{\alpha} K_{\alpha}^{(d)} L^{(d)} L^{(d)\dag} K_{\alpha}^{(d)\dag}\\
        & = \sum_{\alpha=1}^A  J_{\alpha}^{(d+1)} J^{(d+1)\dag}_{\alpha}
        = L^{(d+1)} L^{(d+1)\dag}
    \end{split}
\end{equation}
where $J_\alpha^{(d+1)} \equiv \sqrt{p_{\alpha}} K_{\alpha}^{(d)} L^{(d)}$, $\mathcal{K}^{(d)}$ is a Kraus operator, and $K^{(d)}_\alpha$ and $p_\alpha$ are its corresponding Kraus matrices and probability factor, and $A$ is the number of Kraus matrices in the operation. $L^{(d+1)}$ is formed by concatenating $J_\alpha^{(d+1)}$ as columns, ie. $L^{(d+1)} \equiv [J^{(d+1)}_1, J^{(d+1)}_2, ..., J^{(d+1)}_A]$.  Note that, due to the concatenation, the number of columns of $L$ changes after each noise operation; each column in $L^{(d)}$ will evolve to $A$ columns in $L^{(d+1)}$. For example, if the dimension of $L^{(d)}$ is $2^N \times 3$ and $A=2$, then $L^{(d+1)}$ has dimension  $2^N \times 6$.

\subsection{Eigenvalue Truncation} \label{sec:eigenvalue truncation}
From ($d-1$)-th layer to $(d)$-th layer of a quantum circuit, the number of $J_\alpha$ vectors grows by $A$ times. For a system starting from a pure state, this number is $A^d$ at the $d$-th layer. This scaling makes tracking all $J_\alpha$ vectors computationally intractable over time if left unchecked. Furthermore, in practice, when the noise level is small, the number of columns corresponding to significant eigenvalues of $L^{(d)}$ is often found to grow only polynomially with the system size. An eigenvalue truncation procedure is used to project the density matrix into lower rank and keep only those highest contributing columns, akin to the quantum filtering in simulating open quantum system \cite{Handel2005}. We truncate those eigenvectors whose eigenvalues are negligible by 
\begin{equation}
\rho^{(d)} = U^{(d)} \Lambda^{(d)} U^{(d)\dag}
\simeq \tilde{U}^{(d)} \tilde{\Lambda}^{(d)} \tilde{U}^{(d)\dag}
\end{equation}
where $U^{(d)}$, $\Lambda^{(d)}$ are eigenvectors and eigenvalues of $\rho^{(d)}$, and from which we define $\tilde{U}^{(d)}$, $\tilde{\Lambda}^{(d)}$ as approximations for which the unimportant eigenvalues and eigenvectors are truncated. The truncation is based on a threshold $\epsilon$. The descending-ordered eigenvalues are picked up one by one until they sum to $1-\epsilon$. The remaining eigenvalues sum to $\epsilon$ are thrown away along with their associated eigenvectors. Although more sophisticated ways of truncation exist \cite{Alquier2013,Butucea2015}, we use this simple cutoff criteria to better control the error introduced by the procedure. Furthermore, this truncation method is the optimal scheme to preserve the trace and the 2-norm of a matrix, known as the Eckart-Young theorem \cite{Trefethen1997}.

The representation above might be a useful method to retain only the maximal information $L$ factors in Kraus noise models. However, the approach appears to have the computational problem that it involves the eigendecomposition of a $2^N \times 2^N$ matrix. Solving an eigenvalue problem of a $2^N\times 2^N$ matrix has a complexity of $\mathcal{O}((2^N)^3)$, which is very expensive and would overwhelm the benefit of low rank simulation. However, using the result from theorem 1 in the appendix, we can efficiently compute the eigenvectors and eigenvalues without explicitly constructing the density matrix. The complexity of the eigenvalue problem is instead $\mathcal{O}((A V)^3)$ where $V<2^N$ is the number of columns of $L$ in Eqn. (\ref{eqn:low rank decomposition}), and $A$ is the number of Kraus matrices comprising the Kraus operator.

\subsection{Kraus Operator Decomposition} \label{sec:kraus operation decomposition}
We model noisy quantum channels with single-qubit Kraus operators. To model noise induced by gate operations, one may apply Kraus operators following each gate. If the circuit is sparse in terms of gates, correspondingly, there are only a few Kraus operators per layer. In this case, $A V$ stays small and eigenvalue truncations can be done relatively efficiently. However, consider a dense noisy circuit with a depolarizing Kraus operator acting on every qubit at each time-step (Fig. \ref{fig:general circuit}). The depolarizing Kraus operator is comprised of $4$ matrices \cite{Nielsen2011}, and therefore $A = 4^N$, making eigenvalue truncation intractable with complexity $\mathcal{O}((4^N V)^3)$. This can be resolved by decomposing the Kraus operator in several groups 
\begin{equation}
    \mathcal{K}^{(d)} = \prod_{\beta=1}^B \mathcal{K}_\beta^{(d)}.
\end{equation}
This decomposition is possible because noise channels in a quantum computer can be well described by a combination of one and two-qubits Kraus operators \cite{Arute2019}. In the example in Fig. \ref{fig:general circuit}, instead of an eigenvalue truncation after each whole Kraus layer, a truncation is applied after each $\mathcal{K}_\beta^{(d)}$ in order to prevents $A$ from getting too large. In other word, the evolution $\rho^{(d+1)} \rightarrow \mathcal{K}^{(d)}\rho^{(d)}$ is approximated by 
\begin{equation}
    \rho^{(d+1)}
    =\mathcal{E}\mathcal{K}_B^{(d)}\mathcal{E}\mathcal{K}_{B-1}^{(d)}\cdots\mathcal{E}\mathcal{K}_2^{(d)}\mathcal{E}\mathcal{K}_1^{(d)}\rho^{(d)}
\end{equation}
where $\mathcal{E}$ is eigenvalue truncation operation. Because a Kraus operator is decomposed into $B$ operations, each decomposed operation has only $\bar{A}=4^{\frac NB}$ Kraus matrices, denoted as $K^{(d)}_{\beta,\alpha}$ with subscript $\alpha$ runs from $1$ to $\bar{A}$. Each truncation has complexity $\mathcal{O}((4^{\frac{N}{B}}V)^3)$ and there are $B$ of truncation steps, so the total complexity is $\mathcal{O}(B(4^{\frac{N}{B}}V)^3)$. We choose $B$ such that $M=\frac NB$ is constant and the complexity becomes $\mathcal{O}(N\frac{(4^MV)^3}{M})$. We define the ``intermediate rank'' $V_I$ as
\begin{equation} \label{eqn:intermediate rank}
    V_I^3 \equiv N\frac{(4^MV)^3}{M}.
\end{equation}
This quantity will be important to estimate the conditions for which low rank simulation is faster than full density matrix simulation in Section \ref{sec:implementation}. 

We focus on the simulation of NISQ-era circuits, which are shallow in terms of circuit depth. However, note that for very deep circuits, $V_I$ can grow larger than $2^N$. In this case, there is no benefit of doing low rank evolution, and we switch back to full density matrix evolution. The algorithm of low rank noise simulation is summarized in Algorithm 1.

\subsection{The LRET Algorithm} \label{sec:probability density and measurement}
Section \ref{sec:low rank evolution}, \ref{sec:eigenvalue truncation} and \ref{sec:kraus operation decomposition} together describe the algorithm for getting the final low rank representation, $L^{(D)}$. We refer this algorithm to as Low Rank simulation with Eigenvalue Truncation (LRET). The full procedure is summarized in Algorithm \ref{alg:low rank evolution}. The concatenation and the eigenvalue truncation in the algorithm are described in Section \ref{sec:low rank evolution} and \ref{sec:eigenvalue truncation}, respectively. Note that although we use the $\ket{000...}$ fudicial state as the initial state (as shown in Algorithm \ref{alg:low rank evolution}) throughout this article, the algorithm works with any single fiducial state or sparse linear combination of states. Also note that density matrix, probability distribution and expectation value in Algorithm \ref{alg:low rank evolution} are optional and are included as examples of user-specified outputs.

Once we have the final low rank representation, $L^{(D)}$, we can construct the density matrix using Eqn. (\ref{eqn:low rank decomposition}). However, this full density matrix quantity is rarely needed in standard practice. For example, one may want to simulate the behavior of quantum hardware where the only information we get is from measurements in a fixed computational basis which sample the probability mass function, $\text{Prob}(x)$, that is defined by the underlying density matrix. In this case, low rank simulation gains an additional speedup as the probability distribution is simply
\begin{equation}
    \text{Prob}(x) = \sum_{v=1}^V L^{(D)}_{x,v} L_{v,x}^{(D)\dag}.
\end{equation}
where subscript $x$ and $v$ run over computational basis dimension and column dimension of the matrix respectively. The measurement count for each state is then sampled from this distribution. Note that, if the goal of a circuit simulation is to observe and count the measurement outputs, a density matrix is not formed at any point of the simulation as long as the intermediate rank is smaller than $2^N$. Similarly, we can evaluate observables in low rank form using $O = \text{Tr}(\rho\mathcal{O})=\text{Tr}(LL^\dagger\mathcal{O})=\text{Tr}(L^\dagger\mathcal{O}L)$ where $O$ is the expectation value of an observable $\mathcal{O}$.

\begin{minipage}{.9\linewidth}
    \begin{algorithm}[H] 
        \begin{algorithmic}
            \State{$L^{(0)} = [1,0,0,...,0]$}
            \For{$d\gets 1$ to $D$}
                \State{$L^{(d)} \gets G^{(d)} L^{(d-1)}$}
                \For{$\beta\gets 1$ to $B$}
                    \State{$L^{(d)} \gets \text{Concatenate}_{\alpha} (\sqrt{p^{(d)}_{\alpha}} K_{\beta,\alpha}^{(d)} L^{(d)}$})
                    \State{$L^{(d)} \gets \text{Eigenvalue Truncation}_\epsilon(L^{(d)})$}
                \EndFor
            \EndFor \\
            \textbf{Compute quantities of user's choice}:
                \State{Density Matrix $\rho^{(D)} \gets L^{(D)} L^{(D)\dag}$}
                \State{$\text{Prob}(x) \gets \sum_{v} L^{(D)}_{x,v} L^{(D)\dag}_{v,x}$}
                \State{$\text{Expectation} \gets \text{Tr}(L^{(D)\dagger} \mathcal{O}L^{(D)})$}
        \caption{Low Rank Simulation with Eigenvalue Truncation}
        \label{alg:low rank evolution}
        \end{algorithmic}
    \end{algorithm}
\end{minipage}


\section{Implementation and Benchmarking} \label{sec:implementation}
We implemented the algorithm for low rank noise simulation in an in-house quantum circuit simulator built in Python. In our simulator, one can specify one of two options for a noisy simulation: full density matrix simulation (FDM), and low rank simulation with eigenvalue truncation (LRET) as described in Algorithm \ref{alg:low rank evolution}. We benchmark the two simulation methods in three scenarios: randomized benchmarking, state preparation for quantum chemistry and Grover's search algorithm. For time benchmarking, we use Cirq 0.5.0, a widely-used open source FDM simulator, to show that our implementation of FDM method is reasonably optimized and serves as a good baseline for comparison. All benchmarking are executed on an AWS c5.12xlarge instance. 

The general result is that the LRET method is two orders of magnitude faster than the FDM method with a trade-off of $\sim0.01\%$ error. The error is measured by the distance between the output density matrices from the LRET method ($\rho_\text{LRET}$) and from an exact method ($\rho_\text{exact}$), such as FDM. Because this quantity depends on the noise level, we define a more appropriate measure for error benchmarking
\begin{equation}
    \text{distortion}
    \equiv \frac{T(\rho_\text{LRET},\rho_\text{exact})}{T(\rho_\text{exact},\rho_\text{noiseless})}
\end{equation}
where $\rho_\text{noiseless}$ is the density matrix from the simulation of the same circuit without noise, and $T$ is the variational distance \cite{Crooks2015} between the probability distributions defined by the two density matrices in the computational basis, ie. $T(\rho^A, \rho^B)=\sum_i |\rho^A_{ii}-\rho^B_{ii}|$. To aid in a qualitative understating of the distortion measure, it is useful to note that the distance between $\rho_\text{LRET}$ and $\rho_\text{exact}$ captures the error or information loss incurred by the eigenvalue truncation procedure. The denominator scales this value relative to the change induced by the noise channel to $\rho_\text{noiseless}$. For example, when the output error is $0.01\%$ and the change induced by noise is $0.1\%$, the distortion is $10\%$.

\subsection{Randomized Benchmarking}
Randomized benchmarking is a standard tool used to evaluate the performance of quantum hardware \cite{Emerson2005,Knill2008,Onorati2019}. We use the idea to benchmark time and error metrics for the two simulation methods on an ensemble of randomly generated circuits. The circuits are generated from random choices of common gates, including X, Y, Z, S, T, RX, RY, RZ, SWAP, CZ and CNOT. The Section \ref{sec:benchmark circuit type} below discusses the different types of circuits we use for benchmarking. If not explicitly stated otherwise, the random circuits are dense circuits where 1-qubit and 2-qubit gates appear with equal probability in $\mathcal{G}$ (Fig. \ref{fig:general circuit}), and the 2-qubit gates connect to adjacent qubits. Dense circuits are those for which a gate acts on each qubit at each time-step. Inspired by the fact that noise is well described by a set of Kraus operators whose dimension does not scale with circuit size \cite{Arute2019}, all Kraus operators act on one qubit in the benchmarking, as illustrated in Fig. \ref{fig:general circuit}.

To understand the conditions for which the LRET method gains speed-up against FDM method, we also inspect the rank evolution of density matrix in LRET. The rank and the intermediate rank (as defined in Eq. (\ref{eqn:intermediate rank})) directly influence the computational complexity and the speed of the LRET algorithm. In the following sections, we first benchmark time, error, and rank of simulations under different noise channels. Then, we assess performance on a variety of differently characterized random circuits.

\begin{figure}[]
    \centering
    \makebox[0pt]{
        \Qcircuit @C=1.2em @R=0.2em {
         \lstick{\ket{0}} & \multigate{4}{\mathcal{G}^{(1)}} & \gate{\mathcal{K}_1} 
                          & \multigate{4}{\mathcal{G}^{(2)}} & \gate{\mathcal{K}_1} 
                          & \multigate{4}{\mathcal{G}^{(3)}} & \gate{\mathcal{K}_1} & \push{\quad ...\quad}\qw\\
         \lstick{\ket{0}} & \ghost{\mathcal{G}^{(1)}} & \gate{\mathcal{K}_2} 
                          & \ghost{\mathcal{G}^{(2)}} & \gate{\mathcal{K}_2} 
                          & \ghost{\mathcal{G}^{(3)}} & \gate{\mathcal{K}_2} & \push{\quad ...\quad}\qw\\
         \lstick{\ket{0}} & \ghost{\mathcal{G}^{(1)}} & \gate{\mathcal{K}_3} 
                          & \ghost{\mathcal{G}^{(2)}} & \gate{\mathcal{K}_3} 
                          & \ghost{\mathcal{G}^{(3)}} & \gate{\mathcal{K}_3} & \push{\quad ...\quad}\qw\\
         \lstick{\ket{0}} & \ghost{\mathcal{G}^{(1)}} & \gate{\mathcal{K}_4} 
                          & \ghost{\mathcal{G}^{(2)}} & \gate{\mathcal{K}_4} 
                          & \ghost{\mathcal{G}^{(3)}} & \gate{\mathcal{K}_4} & \push{\quad ...\quad}\qw\\
         \lstick{\ket{0}} & \ghost{\mathcal{G}^{(1)}} & \gate{\mathcal{K}_5} 
                          & \ghost{\mathcal{G}^{(2)}} & \gate{\mathcal{K}_5} 
                          & \ghost{\mathcal{G}^{(3)}} & \gate{\mathcal{K}_5} & \push{\quad ...\quad}\qw\\
         }
     }
    \caption{Schematic of a noisy quantum circuit. Here $\mathcal{G}^{(i)}$ represents gate operations at the $i$-th layer and $\mathcal{K}$ represents a one-qubit Kraus operator that models the noise.}
    \label{fig:general circuit}
\end{figure}
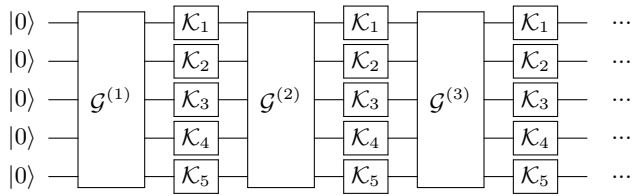

\subsubsection{Depolarizing Noise Channel} \label{sec:benchmark dep}
In this section, we benchmark dense circuits with the depolarizing noise channel, which is defined as $\rho \rightarrow (1-p) \rho + \frac p3 X \rho X + \frac p3 Y \rho Y + \frac p3 Z \rho Z$. It has been shown that noise in quantum hardware, in average, behaves like depolarizing noise \cite{Emerson2002,Weinstein2004}, making it a good description of realistic noise channels.

Fig. \ref{fig:benchmark time (dep)}a shows that, while our FDM method and Cirq take a similar amount of time to run for 12-qubits circuits with $p=0.1\%$ under depolarizing noise, the LRET method is much faster than both. In shallow circuits, LRET is $200\times$ faster than FDM (Fig. \ref{fig:benchmark time (dep)}b). Even for higher depth circuits, LRET remains roughly $100\times$ faster. This can be understood by considering the size of the numerical representation these methods are keeping track of. While the FDM method evolves a $2^N\times2^N$ density matrix, LRET only keeps track of a $2^N\times V$ representation of a density matrix. Effective use of the LRET algorithm amounts to choosing the truncation threshold as to best manage the trade-off between the speed of the simulation and the error in the simulation results. This trade-off is characterized in Fig. \ref{fig:benchmark error (dep)}.

Although $V$ is always smaller than $2^N$ at low depth for $N>2$ (Fig. \ref{fig:benchmark time (dep)}c), the conditions for a speed-up is determined by the intermediate rank $V_I$ defined in Eq. (\ref{eqn:intermediate rank}); the LRET method is faster if $V_I< 2^N$. As shown in Fig. \ref{fig:benchmark time (dep)}d, a speed-up is only achieved for $N>7$ for the circuit depth consider herein. Since $V_I$ increases approximately polynomially and $2^N$ increases exponentially in $N$, the range of depths for which LRET has an advantage will increase even more as the number of qubits increases. Critically, LRET has an advantage precisely in the range where classical simulations begin to become burdensome. Furthermore, the space of circuit sizes in which LRET provides a significant advantage also characterizes the circuits of the early NISQ area, with few tens of qubits, circuit depth and with noise strength $p<0.01$ \cite{Preskill2018}. 

\begin{figure}[]
  \includegraphics[width=1.0\linewidth]{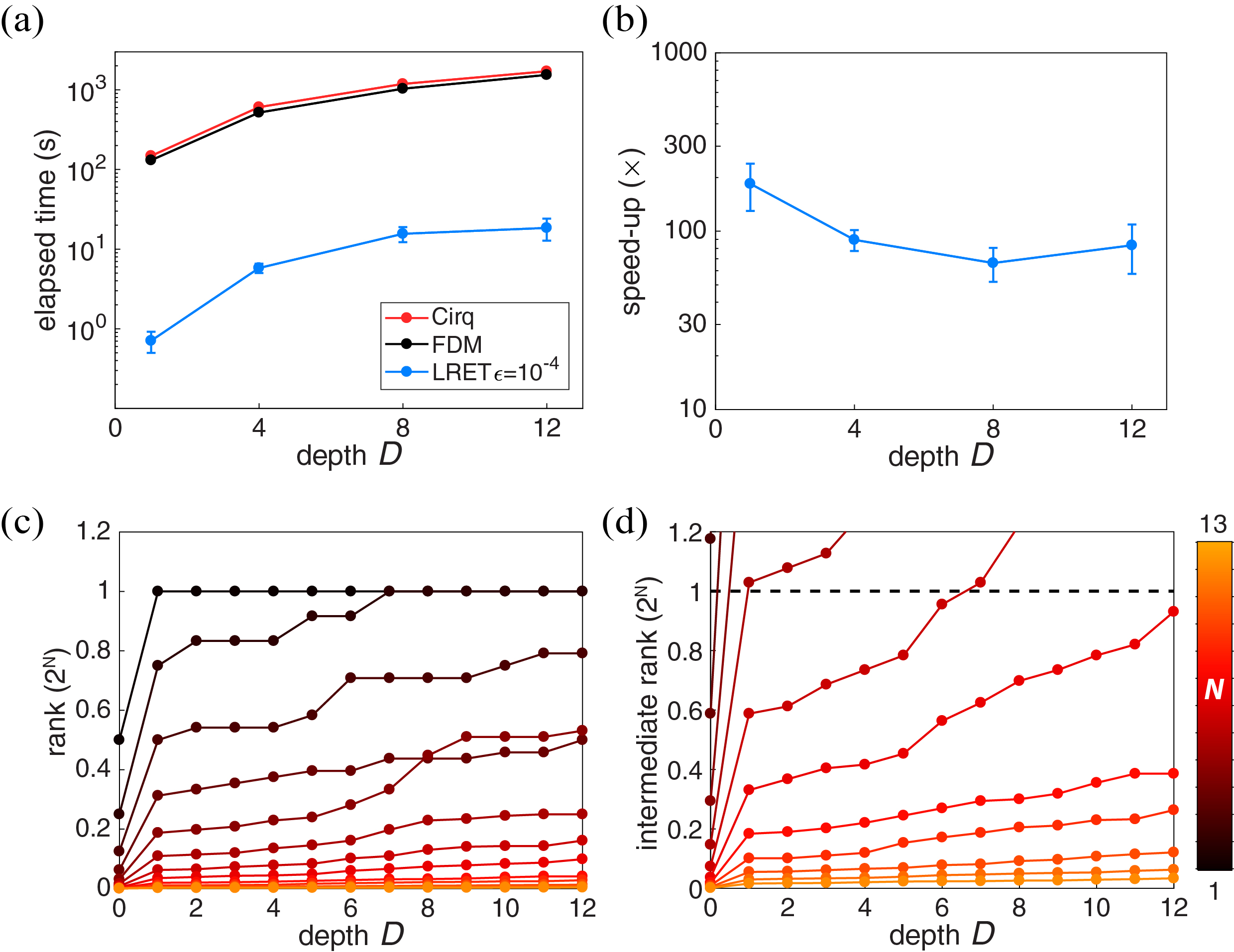}
  \caption{Time benchmarking under depolarizing noise ($p=0.1\%$, $\epsilon=10^{-4}$). (a) Averaged elapsed time for quantum circuits with $N=13$ qubits using different simulators. (b) Speed-up, defined by the ratio of the elapsed time between FDM and LRET. (c) Rank evolution for different size of quantum circuit. (d) Intermediate rank for different size of quantum circuit.}
  \label{fig:benchmark time (dep)}
\end{figure}

The LRET method gains a speed-up by truncating the negligible components of a density matrix. Although the truncation in each step is small, over time the discrepancy from the exact methods, like FDM, can build up. Here, we benchmark the error introduced by the eigenvalue truncation in the LRET method. 

Fig. \ref{fig:benchmark error (dep)}a shows the distortion as a function of the number of qubits ($N$), depth of the circuit ($D$) and eigenvalue truncation threshold ($\epsilon$). While the distortion depends on $N$ and $D$, $\epsilon$ is the strongest factor. There is a general trend that the error starts to increase rapidly at $\epsilon \simeq 10^{-4}$, so we take $\epsilon = 10^{-4}$ as a reasonable choice. From Fig. \ref{fig:benchmark error (dep)}b, we can see that error $<8\%$ for all the $N$ and $D$ considered herein. As we slice out the number of qubits and circuit depths axes in  Fig. \ref{fig:benchmark error (dep)}b, we can see that the error grows roughly linearly with $N$ and $D$ (Fig. \ref{fig:benchmark error (dep)}c and d). 

\begin{figure}[]
  \includegraphics[width=1.0\linewidth]{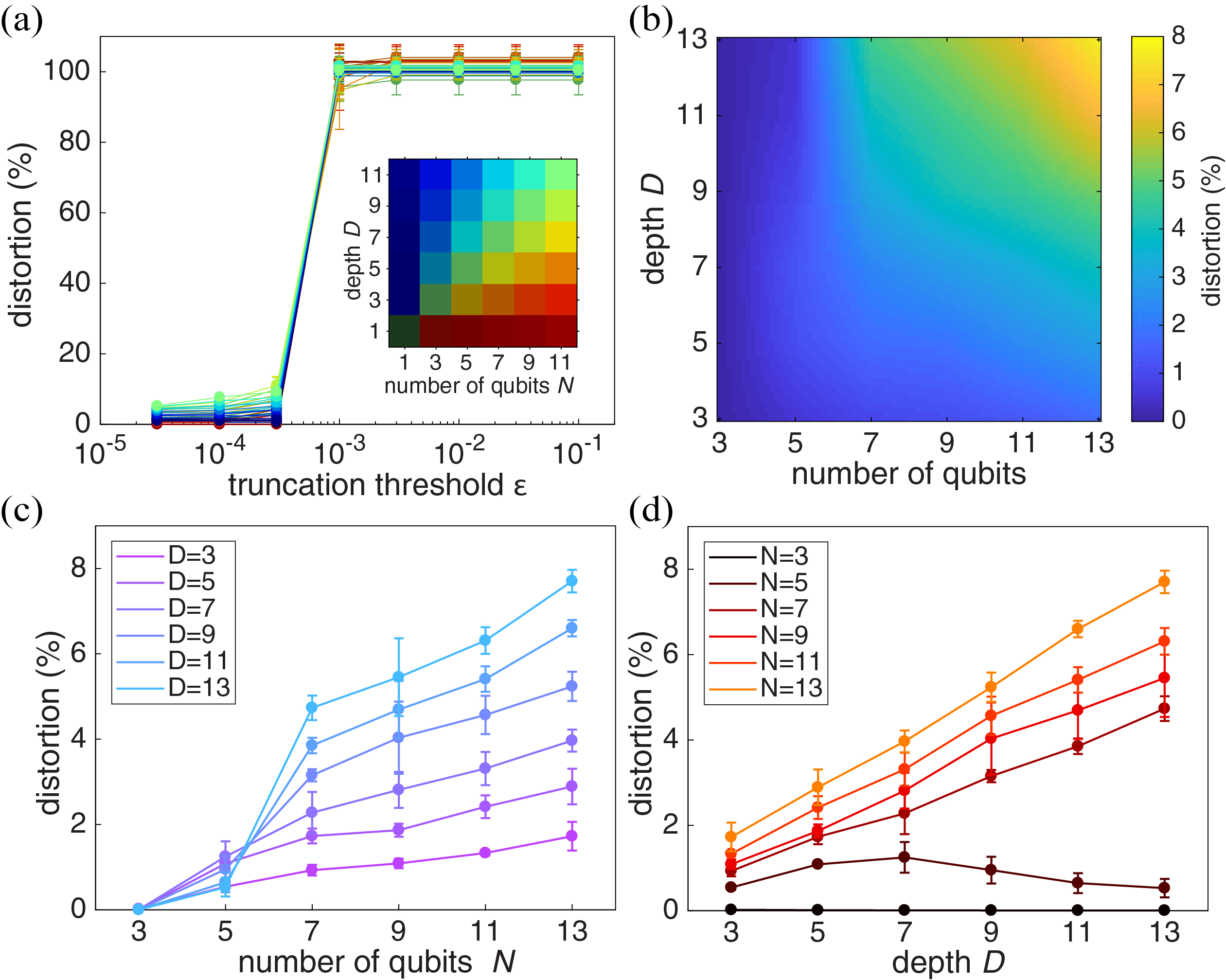}
  \caption{Error benchmarking under depolarizing noise ($p=0.1\%$). (a) Distortion as a function of number of qubits $N$, circuit depth $D$, and $\epsilon$. (Inset) The colormap for the curves, indicating their circuit size. (b) Distortion as a function of $N$ and $D$ with $\epsilon=10^{-4}$. (c) Horizontal line cut of b. (d) Horizontal line cut of b.}
  \label{fig:benchmark error (dep)}
\end{figure}

\subsubsection{Noise Strength} \label{sec:benchmark noise strength}
We now see how the noise strength affects the performance of the LRET method. All benchmarking in this section uses dense circuits with depolarizing noise channels of various strengths. In Fig. \ref{fig:benchmark gamma (dep)}a, the speed-up of the LRET method against the FDM method degrades as the noise strength grows. From $p=0.1\%$ to $p=1\%$, the speed-up drops from the order of $100\times$ to $10 \times$ at $D=12$. This degradation is due to the higher order terms in noise which scale super-linearly in $p$. While the truncation threshold $\epsilon$ is adapted linearly by fixing the ratio $\epsilon/p=0.1$ in this benchmarking, more higher order terms need to be included to meet the truncation threshold. This results in a larger $V$ and $V_I$ (Fig. \ref{fig:benchmark gamma (dep)}b), and thus a longer computational time.

Fig. \ref{fig:benchmark gamma (dep)}c shows that the distortion as a function of $\epsilon$ has a universal shape regardless of the circuit size ($N\times D$) and/or noise strength $p$. The magnitude of the distortion is relatively insensitive to circuit size. The noise strength $p$ proportionally shifts the curves in $\epsilon$ axis (ie. when $p$ and $\epsilon$ are scaled by a same factor, the error stays in a similar range).

\begin{figure}[]
  \centering
      \includegraphics[width=1.0\linewidth]{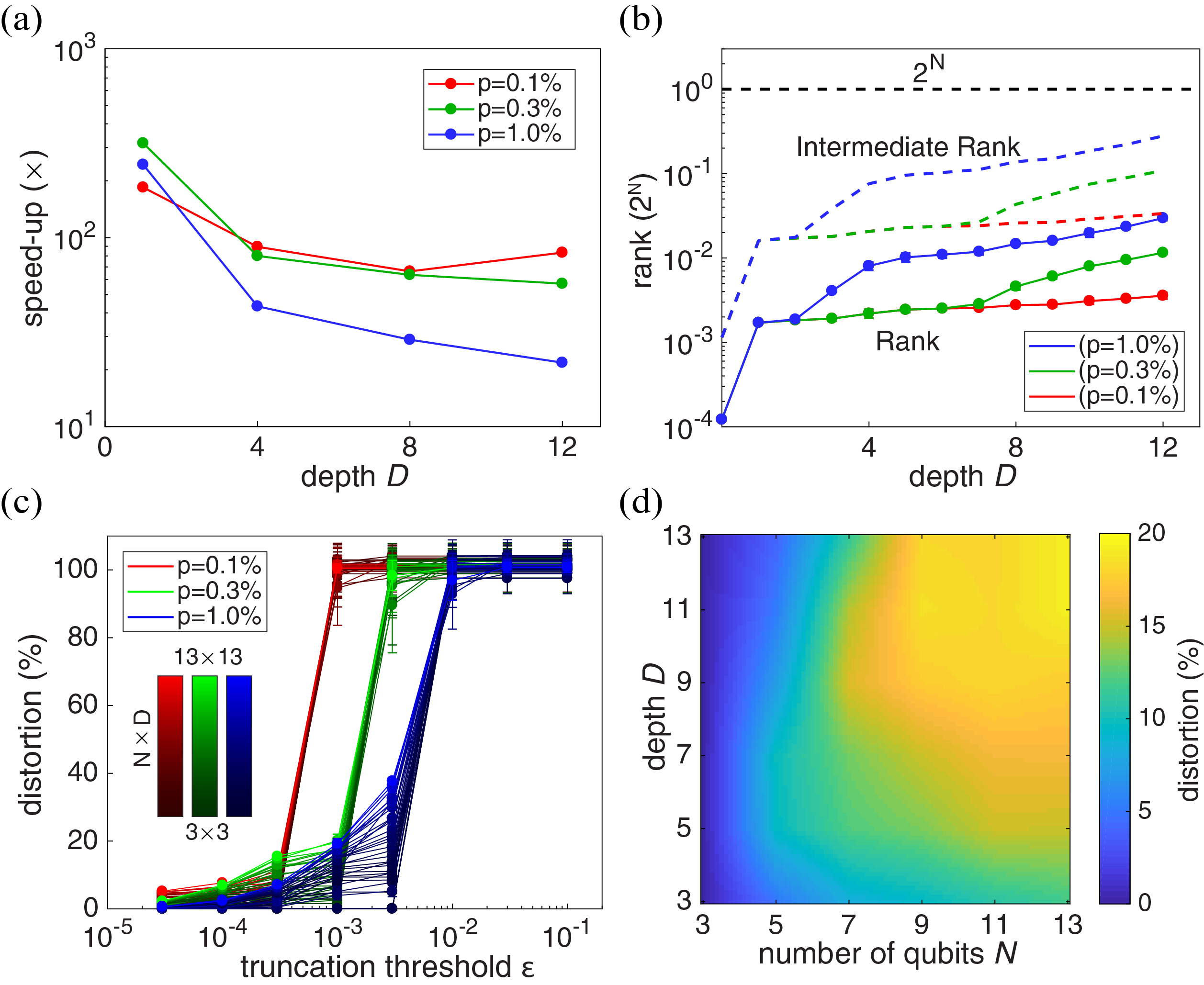}
  \caption{Noise strength benchmarking under depolarizing noise. (a) Distortion as a function of circuit size and eigenvalue truncation threshold $\epsilon$ for noise strength $p=0.1\%$, $p=0.3\%$ and $p=1.0\%$. (b) Distortion for $p=1\%$ and $\epsilon=10^{-3}$ as a function of $N$ and $D$. (c) Comparison of elapsed time of the LRET method for under the three noise strengths ($N=13$). (d) Comparison of $V$ and $V_I$ for LRET method for the three noise strengths ($N=13$). In (b), (c) and (d), the ratio $\epsilon/p$ stays constant.}
  \label{fig:benchmark gamma (dep)}
\end{figure}

\subsubsection{Other Noise Channels} \label{sec:benchmark noise type}

We now consider noise simulations under bit flip and amplitude damping channels for dense circuits with $p=0.1\%$ and $\epsilon=10^{-4}$. Bit flip can be represented by $\rho \rightarrow (1-p) \rho + p X \rho X$ in the operator-sum formalism. In other words, this channel takes a portion of the quantum state and project it uniformly in the $X$ direction of the Hilbert space. Bit flip channel is a special case of anisotropic noises. The results for the bit flip channel generalizes to other types of anisotropic noise in randomized benchmarking, such as phase flip and all other channels described by $\rho \rightarrow (1-p) \rho + p U \rho U^{\dag}$ where $U$ is any $2\times2$ unitary matrix. 
The amplitude damping channel dissipates the energy of a qubit towards its lower energy basis, usually denoted as the $\ket{0}$. We use the operator-sum formalism $\rho \rightarrow E_0 \rho E_0 + E_1 \rho E_1$, where $E_0$ and $E_1$ are Kraus operators for amplitude damping, defined as $E_0 = 
	\begin{bmatrix}
	1 & 0 \\
	0 & \sqrt{1-p}
	\end{bmatrix}$
and 
$E_1 = 
	\begin{bmatrix}
	0 & \sqrt{p} \\
	0 & 0
	\end{bmatrix}.$

Fig. \ref{fig:benchmark time (noise type)}a shows that LRET is significantly faster than FDM for all noise types. Especially in amplitude damping channel, where LRET completes in less than 3 seconds while FDM takes more than 25 minutes at $\text{depth}=12$. The speed-up of the depolarizing and bit flip channels are about $100\times$ faster while the speed-up for amplitude damping is almost $1000\times$ faster (Fig. \ref{fig:benchmark time (noise type)}b). This is related to the slower increase of intermediate rank (Fig. \ref{fig:benchmark time (noise type)}c) due to the fact that amplitude damping has a preferred state, the $\ket{0}$ state, regardless of the details of the qubit state.  

\begin{figure}[]
      \includegraphics[width=1.0\linewidth]{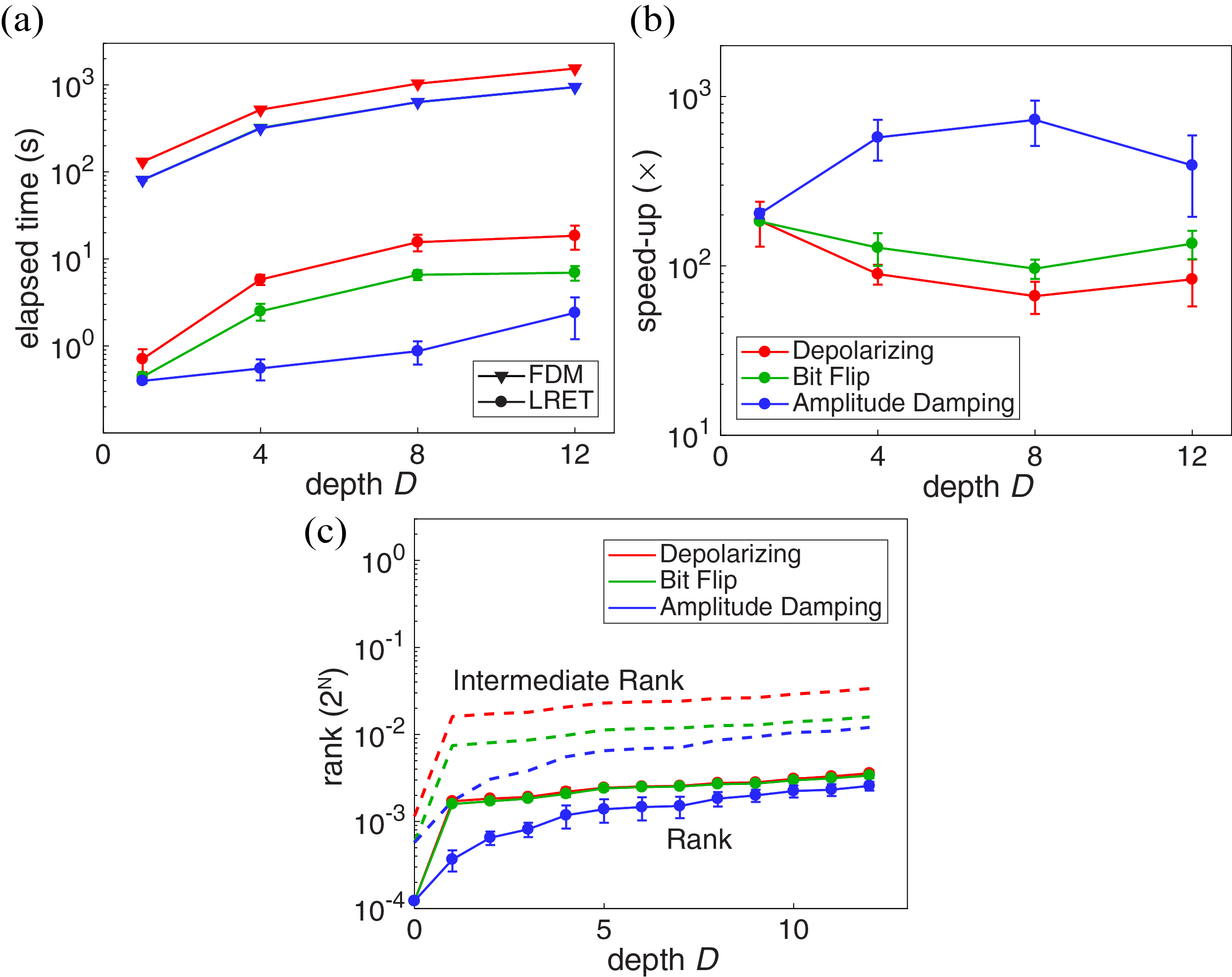}
  \caption{Noise type benchmarking: Time ($p=0.1\%$, $N=13$, $\epsilon=10^{-4}$). (a) Averaged elapsed time using FDM and LRET. The color code is the same as b. (b) The speed-up, defined by the ratio of elapsed time between FDM and LRET. (c) The evolution of rank (solid line) and intermediate rank (dashed line) for different $D$.}
  \label{fig:benchmark time (noise type)}
\end{figure}

The error benchmarking for the bit flip channel (Fig. \ref{fig:benchmark error (noise type)})a-d is very similar to that of depolarizing channel, except that bit flip is more tolerant to $\epsilon$ when the circuit size is small. At $\epsilon=10^{-3}$, the distortion is $\sim50\%$ in bit flip channel while it is $\sim100\%$ in depolarizing channel. When $\epsilon=10^{-4}$, distortion is reasonably small for all $N$ and $D$ considered herein, so we take $10^{-4}$ as a recommended choice of $\epsilon$.

In contrast to depolarizing and bit flip channels, under amplitude damping the distortion saturates at lower values when $\epsilon$ is large (Fig. \ref{fig:benchmark error (noise type)}e). This is possibly because amplitude damping prefers the ground state, favoring the LRET method which keeps only few important components of a quantum state. When $\epsilon=10^{-4}$, the distortion is smaller than $4\%$ for all circuits considered in Fig. \ref{fig:benchmark error (noise type)}f. The distortion grows slowly but linearly with circuit depth (Fig. \ref{fig:benchmark error (noise type)}g and h).

\begin{figure}[]
  \centering
      \includegraphics[width=1.0\linewidth]{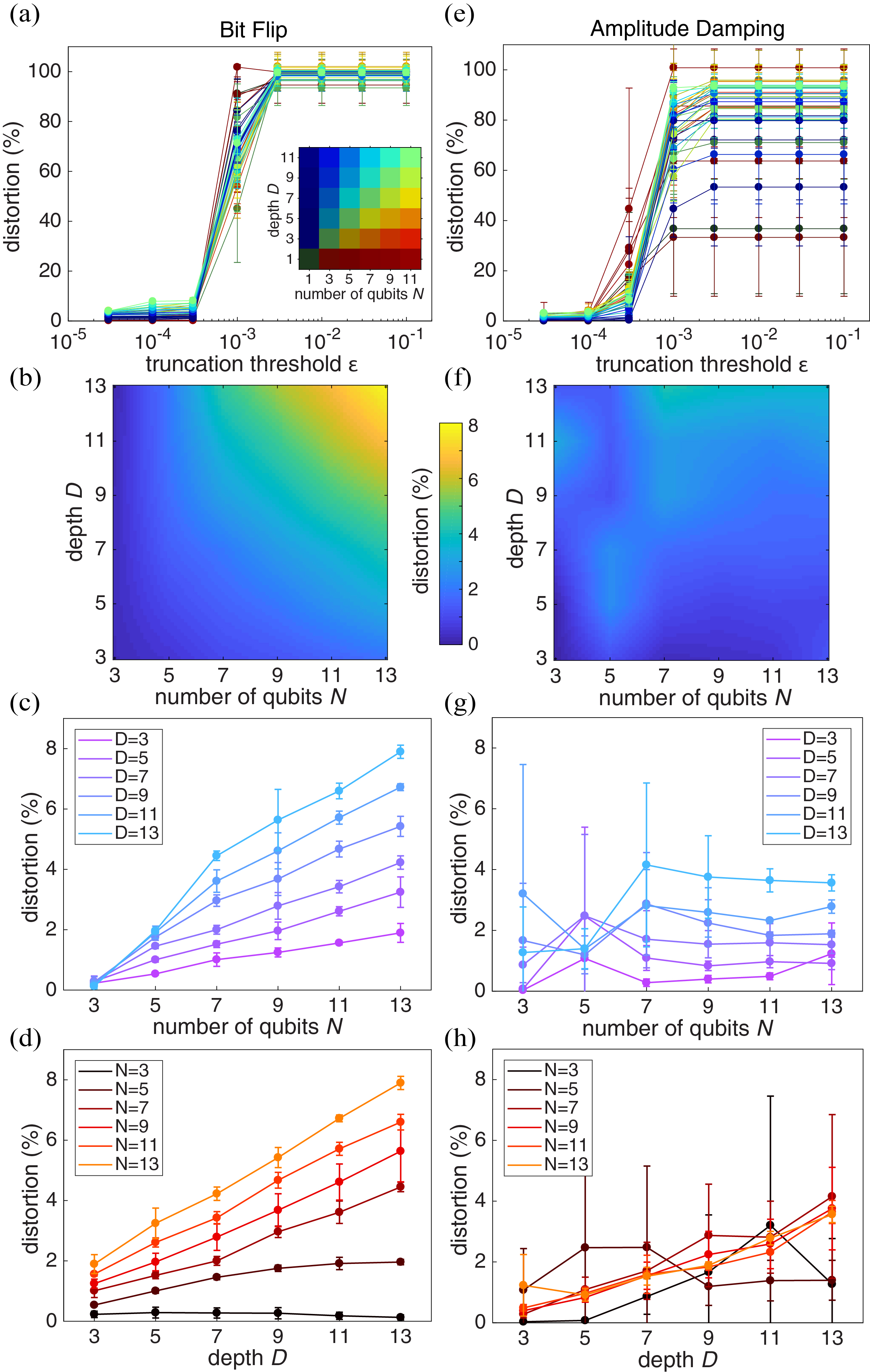}
    \caption{Error benchmarking under bit flip (left column) and amplitude damping (right column) ($p=0.1\%$). (a,e) Distortion as a function of the number of qubits, circuit depth, and $\epsilon$. (Inset) The colormap for the curves, indicating their their sizes of quantum circuit. (b,f) Distortion as a function of qubit number, circuit depth with $\epsilon=\times10^{-4}$. (c,g) Horizontal line cut of b and f. (d,h) Vertical line cut of b and f.}
  \label{fig:benchmark error (noise type)}
\end{figure}

\subsubsection{Sparsity and Connectivity of Quantum Circuits} \label{sec:benchmark circuit type}
A quantum circuit can be characterized by its sparsity and connectivity of gates. In all of the above benchmarking, the random circuits are dense, i.e. for all time steps a gate acts on each qubit (e.g. the circuit in Fig. \ref{fig:general circuit}a), and the connections are local, which means that two-qubit gates only connect adjacent qubits. Below, we consider other types of circuits. The first is dense and global, and the second is sparse and local. In sparse circuits, each qubit does not always interact with a gate at every time step and Kraus operators are only inserted after gates (i.e. the noise is as sparse as the gates). In a globally connected circuit, two-qubit gates can connect any pair of qubits in a circuit. In this section we use the bit flip channel for benchmarking the LRET method on all the aforementioned circuit-types with $p=0.1\%$ and $\epsilon=10^{-4}$. 

In terms of time-cost, simulating different circuit types goes from harder to easier as: dense-global $\rightarrow$ dense-local $\rightarrow$ sparse-local. In Figure \ref{fig:benchmark circuit type}a one can see that the LRET method retains it's speed-up for all circuit types.  The time difference between the sparse and the dense is because there are about twice as many gates in dense circuits than in sparse circuits. The time difference between the global and the local is because the set of all fixed-depth globally connected circuits spans a larger Hilbert space. Therefore, the rank of dense-global circuits grows slightly faster (Fig. \ref{fig:benchmark circuit type}c) and the simulations are slightly slower (Fig. \ref{fig:benchmark circuit type}a). Interestingly, in Fig. \ref{fig:benchmark circuit type}a one can observe that in FDM simulation, due to the number of gates and the memory allocation, the time-complexity grows at a similar rate as that of LRET with increasing circuit depth. Thus, speed-ups are consistently $\sim 100 \times$ faster with LRET than with FDM (Fig. \ref{fig:benchmark circuit type}b). 

\begin{figure}[]
  \centering
      \includegraphics[width=1.0\linewidth]{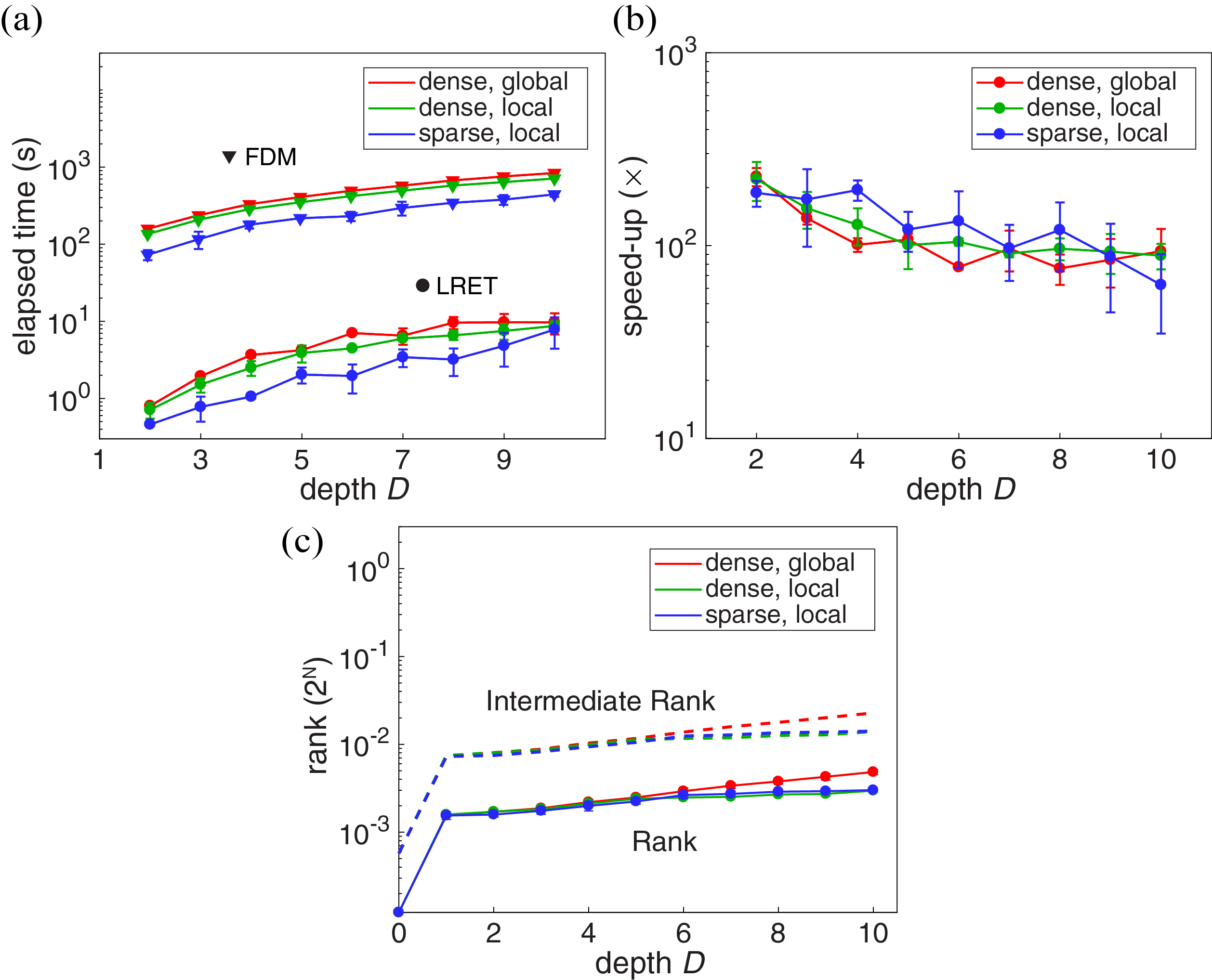}
  \caption{Time benchmarking for different circuit-types under bit flip noise ($p=0.1\%$, $N=13$, $\epsilon=10^{-4}$). The distortion values for each of the circuits simulated using the LRET method are similar and are all less than $9.2\%$. (a) Averaged elapsed time for quantum circuits with different gate sparsity and connectivity. For the sparse circuit, the sparsity is $50\%$. (b) The speed-up of LRET for different circuit-types. (c) The evolution of rank (circle) and intermediate rank (dashed line) of density matrix for different circuit-types.}
  \label{fig:benchmark circuit type}
\end{figure}

\subsection{State Preparation for Quantum Chemistry} \label{sec:benchmark MCVQE}
Quantum simulation is one of the most promising application areas for NISQ devices \cite{Preskill2018}. Algorithms have been developed to solve optimization and physical problems through quantum simulation approaches \cite{Abrams1997,Peruzzo2014,Farhi2014}. Here, we use our low rank noise simulator to run a circuit that generates generalized-amplitude W states \cite{Diker2016} and Dicke states \cite{Bartschi2019}. Although states of this kind are not hard to simulate classically, they are commonly used as subroutines for quantum information processing \cite{Murao1999,Prevedel2009,Pezze2018,Parrish2019,Hadfield2019} and thus are of high interest for simulations. Ordinarily the parameters of the circuit are initialized according to the solution of a configuration interaction singles chemistry problem, but for benchmarking purposes we set the parameters randomly.

Unlike most of the circuits used in randomized benchmarking above, the circuit in Fig. \ref{fig:benchmark chem (dep)}a is sparse. Two ways to model noise channels are (1) placing Kraus operators only after each gate, and (2) after each qubit at every time-step. We call the former sparse noise, and the later dense noise. We use a depolarizing noise channel with $p=0.1\%$ and $\epsilon=10^{-4}$. Figs. \ref{fig:benchmark chem (dep)}b and d show that the LRET method has at least a $10\times$ speed-up, and more than $100\times$ when $N$ is larger. The distortion caused by LRET is about $5\%$ in sparse noise and $15\%$ in dense noise (Fig. \ref{fig:benchmark chem (dep)}c and e). In the 13-qubit circuits, the rank of the final density matrix in LRET is $0.4\%$ and $1\%$ of the full rank in the sparse and dense noise cases, respectively. This is because the rank of a density matrix in a circuit model with dense noise increases faster, and the higher order terms thrown out by eigenvalue truncation become more important. The quantum states produced by this state preparation are highly entangled. The fact that the low rank method gains an order of magnitude speed-up demonstrates its utility when applied to practical algorithms. 

\begin{figure}[]
    \centering
    \subfloat[]{
    \Qcircuit @C=1em @R=0em{
        &\gate{R_y(\theta_1)}&\ctrl{1}            &\targ    &\qw                 &\qw      &\qw                 &\qw      &\qw\\
        &\qw                 &\gate{F_y(\theta_2)}&\ctrl{-1}&\ctrl{1}            &\targ    &\qw                 &\qw      &\qw\\
        &\qw                 &\qw                 &\qw      &\gate{F_y(\theta_3)}&\ctrl{-1}&\ctrl{1}            &\targ    &\qw\\
        &\qw                 &\qw                 &\qw      &\qw                 &\qw      &\gate{F_y(\theta_4)}&\ctrl{-1}&\qw\\
     }} \\
     \subfloat{ \label{fig:}
     \includegraphics[width=1.0\linewidth]{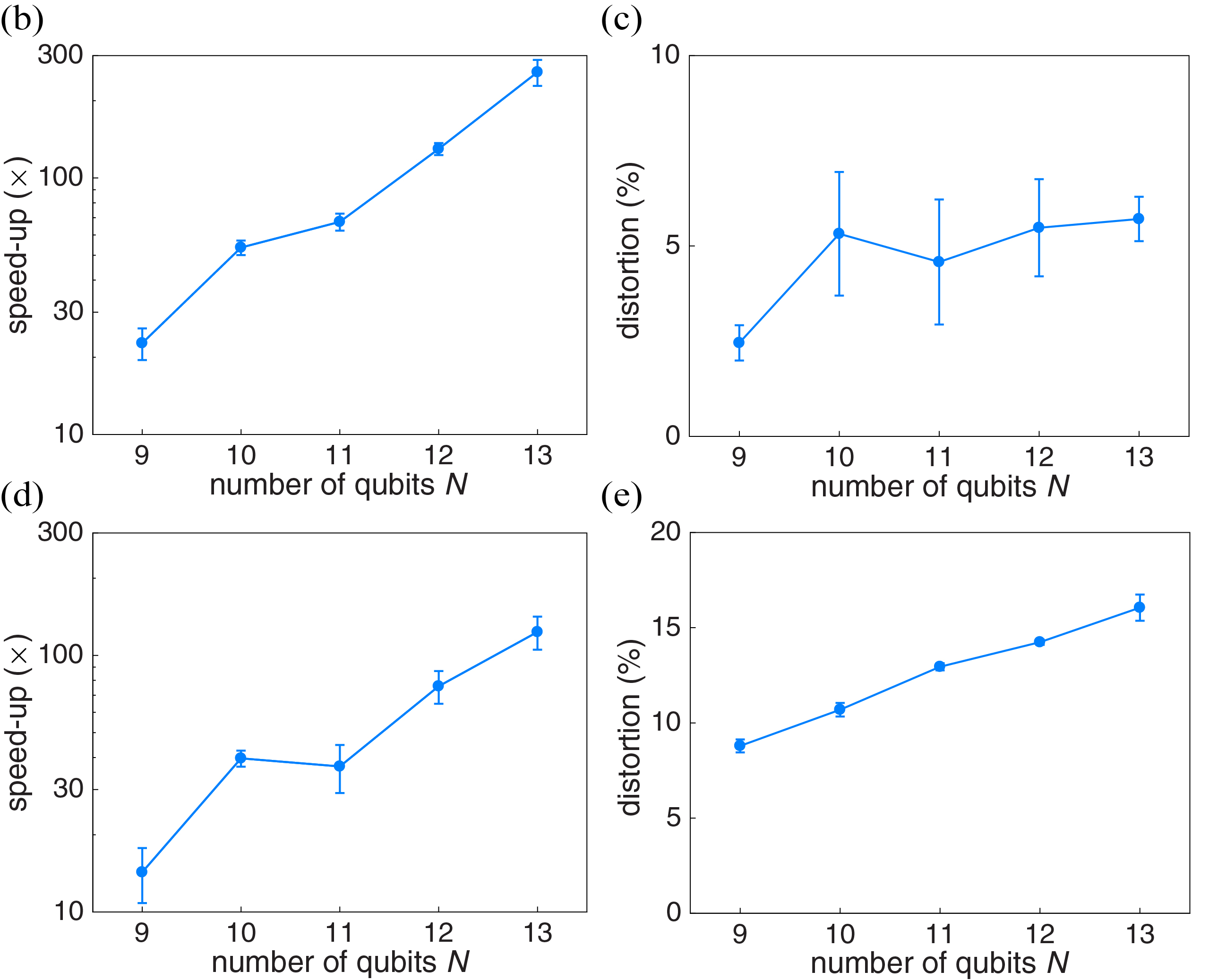}
     }
    \caption{Results of MCVQE state preparation under depolarizing noise ($p=0.1\%$) and $\epsilon=10^{-4}$. (a) Circuit for state preparation with 4 qubits as an example. (b and c) Speed-up and distortion of the LRET method under sparse noise channels. (d and e) Speed-up and distortion of LRET under dense noise channels.}
    \label{fig:benchmark chem (dep)}
\end{figure}

\subsection{Grover's Search Algorithm and Amplitude Amplification} \label{sec:benchmark Grover}

Amplitude amplification is a generalization of Grover's quantum search algorithm \cite{Grover1996,Grover1998,Tulsi2008}. The algorithm aims to find a solution, $x$, such that $f(x) = 1$ if $x$ is a solution and $f(x) = 0$ otherwise, implying $x$ is not a solution. If $x$ is a solution we say that $x$ is \emph{good}. If one was to randomly sample from a search space then $p_\text{good}$ is the probability of sampling a solution. For a classical search algorithm, it is expected that one would have to sample from the input space on the order of $\frac{1}{p_\text{good}}$ times to find a solution; however, using amplitude amplification, one can expect to find a solution using in only $\mathcal{O}(\sqrt{\frac{1}{p_\text{good}}})$ samples --- a quadratic speed-up over the classical case \cite{Grover1996,Brassard1998}.

In the algorithm, qubits are initialized to the a uniform superposition over the entire search space, where each basis state in the superposition corresponds to an element, $x_i \in \mathcal{X}$, where $\mathcal{X}$ is the search space of the problem. Next, a number of unitaries, known as Grover Iterates, act on the initialized state and boost the amplitudes of states that correspond to \emph{good} solutions. The number of Grover Iterates to apply is given by $\floor{\frac{\pi}{4}\sqrt{\frac{1}{p_\text{good}}}}$. A full measurement of the resulting circuit yields states corresponding to \emph{good} solutions with high probability.  

In our implementation, we define the function $f$ such that $f(x)=1$, i.e. a \emph{good} input, when the binary string representation, $x$, has a Hamming Weight, $\text{HW}(x)$, less than or equal to 2.  

\begin{equation}
    f =
    \begin{cases*}
        1 & if HW(x) $\leq$ 2 \\
        0        & otherwise
    \end{cases*}
\end{equation}
We run amplitude amplification on circuits ranging from 9 to 13 qubits with depolarizing noise with $p=0.1\%$ and $\epsilon=10^{-4}$ and compare LRET and FDM methods (Fig. \ref{fig:benchmark AA (dep)}). In the 13-qubit circuit, the rank of final density matrix of LRET is $0.5\%$ of the full rank with a trade-off of $3.7\%$ distortion. The similarity of the measurement results from both methods demonstrates the accuracy of LRET; in other words, that any information loss from eigenvalue truncation is insignificant when sampling from the resulting density matrix. Time-benchmarking of the two methods illustrates the speed-up provided by LRET, which continues to improve as the number of qubits increases.

 We note that it is not the intent of this study to most accurately predict the results of running this experiment on a particular hardware specification. When running on quantum hardware, gates must be decomposed into the set of gates native to the particular hardware, whereas here we model each of the Grover Iterates as a single unitary. Rather, the aim of the study is to show that LRET retains its accuracy and computational advantage not only for the random circuits used for benchmarking but also for circuits that may have legitimate applications and which exhibit more structure than the randomized circuits. 

\begin{figure}[]
  \centering
      \includegraphics[width=1.0\linewidth]{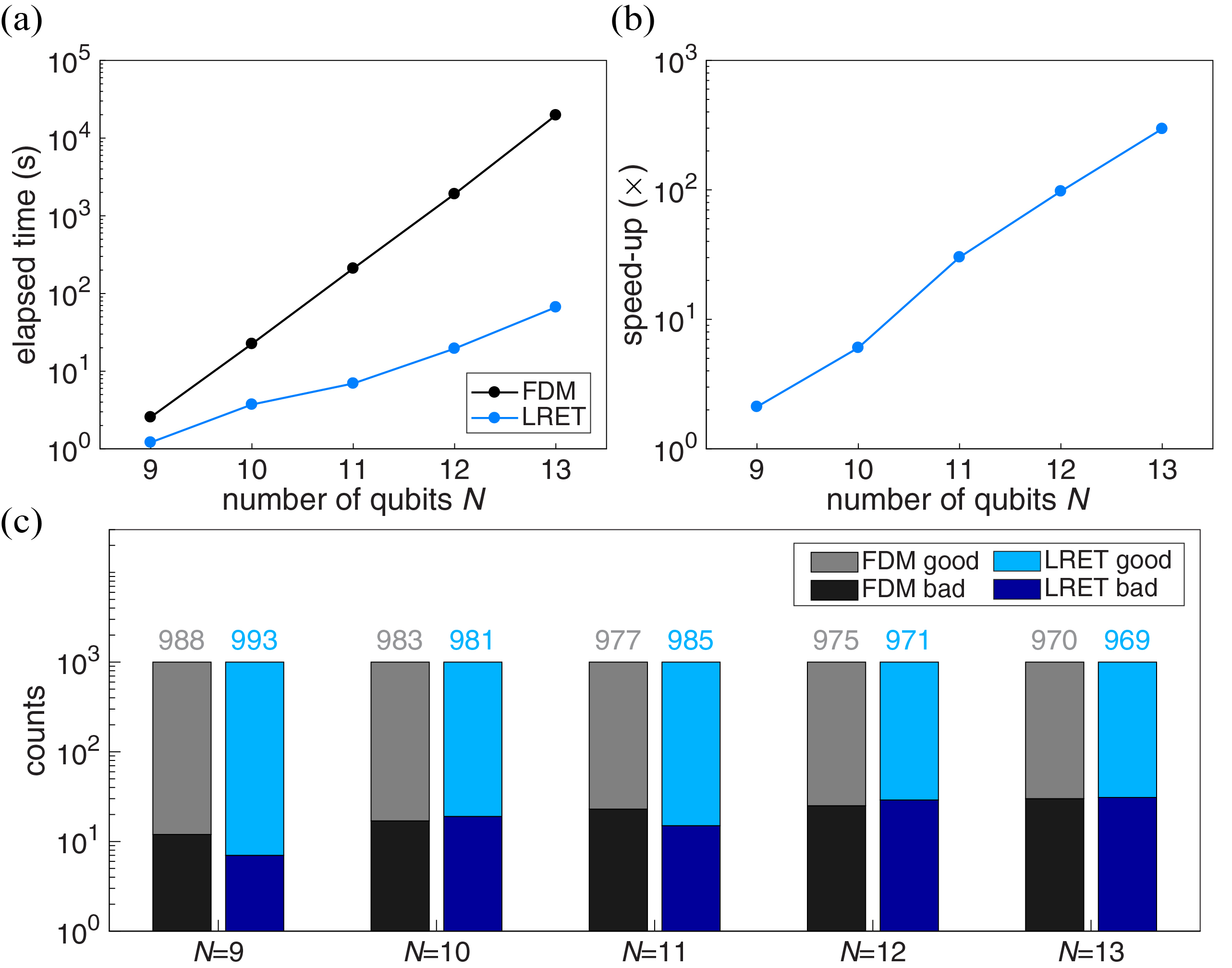}
  \caption{Results of Grover's search algorithm under depolarizing noise ($p=0.1\%$, $\epsilon=10^{-4}$). (a) Elapsed time on a quantum circuit using FDM (black) and LRET (light blue) in log scale. (b) Speed-up, defined by the ratio of elapsed time between FDM and LRET, in log scale. (c) Counts of good and bad solutions. The sampling number is 1000. The number on top of each bar marks the counts of good solution found by Grover's Search algorithm for the corresponding method and number of qubits $N$. The distortion of LRET density matrices are less than $3.8\%$ in these benchmarking.}
  \label{fig:benchmark AA (dep)}
\end{figure}


\section{Summary and Outlook}
In this work we have demonstrated a method to efficiently simulate the evolution of mixed quantum states in noise channels through density matrix decompositions and low rank approximations. Iterative compression of the density matrices enable us to take the advantage of low rank evolution throughout the simulation of a noisy circuit with minimal error. Provided that the noise level of the individual channel is smaller than $0.1\sim1\%$, the density matrices are found to be well approximated by low rank matrices. We provide an entropy argument in the appendix to support this finding. Under the low noise assumption, our results show that the algorithm provides orders of magnitude of speed-up with a small error, on the order of $10^{-4}$, in the probability distributions associated with the output density matrix. The performance in speed and in distortion is robust in different circuit structures and for varying levels of entanglement, since we make no assumption on the symmetry or the entanglement of the circuits. 

We posit that our methodology can be naturally extended to work with observable quantities beyond that of simple measurement probabilities. Furthermore, given that our approach is based in linear algebraic primitives it is likely that the use of GPUs could further improve the performance. While our attention rests on the column space of density matrices, further speed-ups can be achieved by optimizing the representation with respect to the computational basis \cite{Vidal2003,Zhou2020,Noh2020}.

\section{Acknowledgements}
The authors thank Dr. Sean Weinberg, Juan I. Adame, Dr. Fabio Sanches and Dr. Adam Bouland for discussions on the theoretical ground of this work and providing useful references. RMP owns stock/options in QC Ware Corporation.

\appendix

\section{Low Rank Structure in Density Matrices}
\label{adx:LowRankEntropy}
It is found that a density matrix can be well approximated by a low rank matrix when the noise level is low. In this section, we provide an entropy analysis to show that this statement is true for the bit-flip and the depolarizing channels. The quantum circuits considered here have the same structure as Fig. \ref{fig:general circuit}. It is known that the noise in quantum computers can be well characterized by only one and two-qubits Kraus operators \cite{Arute2019}. The result derived here holds for both one and two-qubits Kraus operators when the number of qubits $N$ is large. For convenience, we consider only one-qubit Kraus operators on each qubit. Each Kraus operator $\mathcal{K}$ is assumed to have the form 
\begin{equation} \label{eqn:small noise evolution}
    \begin{split}
        \rho^{(d+1)}=
        \mathcal{K}\rho^{(d)}
        & = \sum^{N_\alpha}_{\alpha=0} p_\alpha K \rho^{(d)} K^\dag \\
        & = (1-p) \rho^{(d)} + \sum^{N_\alpha}_{\alpha=1} p_\alpha K \rho^{(d)} K^\dag
    \end{split}
\end{equation} 
where $K_\alpha$ are Kraus matrices and $\sum^{N_\alpha}_{\alpha=1} p_\alpha=p$. We assume $p$ is small. In other words, the noise level in the circuit is small. 

We use the von Neumann entropy of the density matrix, $S(\rho)=-\text{Tr}(\rho \log_2(\rho))$, to characterize the amount of information in a density matrix, and $R=e^S$ as the effective rank. For a pure state, $S=0$ and $R=1$. For a fully mixed state, $S=N$ and $R=2^N$. Under the noise channels in Eq. (\ref{eqn:small noise evolution}), the entropy is bounded by the property of concavity \cite{Kim2014,Winter2016}
\begin{equation} \label{eqn:concavity}
    \begin{split}
        & \quad 
        \sum_\alpha p_\alpha S(K_{\alpha} \rho^{(d)} K_{\alpha}^\dag) \\
        & \leq 
        S(\rho^{(d+1)})    \\
        & \leq
        \sum_\alpha p_\alpha S(K_{\alpha} \rho^{(d)} K_{\alpha}^\dag) 
        - \sum_\alpha p_\alpha \log_2(p_\alpha) .
    \end{split}
\end{equation}
Note that the von Neumann entropy is non-increasing under the matrix transformation, $K_{\alpha} \rho^{(d)} K_{\alpha}^\dag$, So we have $S(K_{\alpha} \rho^{(d)} K_{\alpha}^\dag) \leq S(\rho^{(d)})$. The equality holds when $K_\alpha$ is an isometry such as unitary matrices. As a result, the second inequality in Eq. (\ref{eqn:concavity}) reduces to 
\begin{equation} \label{eqn:concavity2}
        S(\rho^{(d+1)})   
        \leq
        S(\rho^{(d)}) - \sum_\alpha p_\alpha \log_2(p_\alpha).
\end{equation}
This provides an inequality for entropy change $\Delta S \equiv S(\rho^{(d+1)}) - S(\rho^{(d)})< -\sum_\alpha p_\alpha \log_2(p_\alpha)$. 
The Kraus operators considered here are a direct product of one-qubit Kraus operators. For the case that there is a bit-flip channel on each qubit, $p_\alpha$ follows the Bernoulli distribution with sequence length $N$ and probability $\{1-p,p\}$. The upper bound of the entropy change is $\Delta S_\text{bit}= -\sum_\alpha p_\alpha \log_2(p_\alpha) = -Np\log_2(p)-N(1-p)\log_2(1-p)$. The literature on quantum hardware has reported qubits with noise level of $0.1\%$ \cite{Yang2019,Burrell2010,Zahedinejad,Yoneda2018,Harvey2018,Reed2010}. When the noise is small, $\Delta S_\text{bit}$ is well described by first order approximation in $p$, $\Delta S_\text{bit}=N(p-p\log_2(p))$. The effective dimensionality of the density matrix in this approximation is
\begin{equation}
    \begin{split}
        R_\text{bit} 
        = 2^{\Delta S_\text{bit}}
        = ((\frac{2}{p})^p)^N
        = (1-\gamma_\text{bit})^N
        \simeq 1-N\gamma_\text{bit}
    \end{split}
\end{equation}
where $\gamma_\text{bit}\equiv(\frac{2}{p})^p-1\simeq0.008$ when $p=0.1\%$. The effective dimensionality grows approximately linearly with number of qubits. 

For the case that there is a depolarizing channel on each qubit, $p_\alpha$ follows the categorical distribution with sequence length $N$ and probability $\{1-p,\frac p3,\frac p3,\frac p3\}$. The upper bound of the entropy change is $\Delta S_\text{dep}= -\sum_\alpha p_\alpha \log_2(p_\alpha) = -Np\log_2(\frac p3)-N(1-p)\log_2(1-p)$ or $\Delta S_\text{dep}=N(p-p\log_2(\frac{p}{3}))$ in the first order approximation. The effective dimensionality of the density matrix in this approximation is
\begin{equation}
    \begin{split}
        R_\text{dep} 
        = 2^{\Delta S_\text{dep}}
        = ((\frac{6}{p})^p)^N
        = (1-\gamma_\text{dep})^N
        \simeq 1-N\gamma_\text{dep}
    \end{split}
\end{equation}
where $\gamma_\text{dep}\equiv(\frac{6}{p})^p-1\simeq0.009$ when $p=0.1\%$. The effective dimensionality grows approximately linearly with number of qubits. These results suggest that, while the size of the density matrix grows exponentially with $N$, the effective rank of the density matrix grows linearly with $N$ in a good approximation. The linearly approximation is good with error $<5\%$ up to $N=130$ qubits.

\section{Low Rank Eigendecomposition}
This section provides a theorem for finding eigenvalues efficiently without forming a density matrix explicitly. 
\begin{theorem}
	Let $B$ be a $n \times m$ matrix, where $n>m$. $BB^{\dag}$ and $B^{\dag}B$ share $m$ eigenvalues. Furthermore, if $u$ is a eigenvector of $B^{\dag}B$, $Bu$ is the eigenvector of $BB^{\dag}$ that shares the same eigenvalues.
\end{theorem}
\begin{proof}
	The matrix $B$ has the singular value decomposition 
	$$ B = U\Sigma V^{\dag} $$
	where $U$ is a $n \times m$ matrix with orthonormal columns, and $V$ is a $m \times m$ orthonormal matrix. From $B$, we can construct two Hermitian matrices
	$$ D \equiv BB^{\dag} = U\Sigma V^{\dag}V \Sigma U^{\dag} = U\Sigma^2U^{\dag} \equiv U\Lambda U^{\dag} $$,
	$$ M \equiv B^{\dag}B = V\Lambda V^{\dag} $$
	where $D$ is called full space matrix, and $M$ is the subspace matrix. Denoting an eigenvector of $M$ as $u$, and its eigenvalue as $\lambda$, we have 
	$$ Mu = \lambda u.$$
	Then, by multiplying both sides on the left by $B$ we have
	$$ BMu = BB^{\dag}Bu = D(Bv) = \lambda (Bu).$$
	We have proved that $BB^{\dag}$ and $B^{\dag}B$ share $m$ eigenvalues, and that their eigenvectors, $u_D$ and $u_M$, are related by the equation $u_D = B u_M$. 
\end{proof}

\bibliography{noisemodel}

\end{document}